\documentclass{llncs}
\usepackage{hyperref}       % hyperlinks
\usepackage{url}            % simple URL typesetting
\usepackage{amsfonts}       % blackboard math symbols
\usepackage{nicefrac}       % compact symbols for 1/2, etc.
\usepackage{float}
\usepackage{amsmath} 
\usepackage{mathtools} 
\usepackage{algorithm}% http://ctan.org/pkg/algorithms
\usepackage[noend]{algpseudocode}% http://ctan.org/pkg/algorithmicx
\usepackage{listings}
\usepackage{color}
\usepackage{tikz}
\usepackage[font=small,skip=0pt]{caption}
\usetikzlibrary{arrows,positioning,automata}
\newcommand\alg[1]{{\sc#1}}
\newcommand\tool[1]{{\sc #1}}

\renewcommand\rule[1]{\scalebox{.8}{({\sc \lowercase{#1}})}}
\newcommand\premise[1]{\scalebox{.8}{{\sc \lowercase{#1}}}}
\newcommand\Long[1]{#1}
\newcommand\Short[1]{}
\newcommand\fo{\mathrm{f1}}
\newcommand\ft{\mathrm{f2}}
\newcommand\tone{\mathrm{t1}}
\newcommand\ttwo{\mathrm{t2}}
\newcommand\fosu{\fo^{su}}
\newcommand\ftsu{\ft^{su}}

\title{Regression verification of unbalanced recursive functions with multiple calls \Long{(long version)}}

\author{
  Chaked R.J.~Sayedoff \inst{1}
    \and Ofer Strichman \inst{1}
}

\institute{IE, Technion, Haifa, Israel. \\ 
\email{chaked@campus.technion.ac.il} \hspace{1 cm}	\email{ofers@technion.ac.il}
}
\begin{document}

\maketitle
\begin{abstract}
Given two programs $p_1$ and $p_2$, typically two versions of the same program, the goal of \emph{regression verification} is to mark pairs of functions from $p_1$ and $p_2$ that are equivalent, given a definition of equivalence. The most common definition is that of \emph{partial equivalence}, namely that the two functions emit the same output if they are fed with the same input and they both terminate. 
The strategy used by the Regression Verification Tool (RVT) is to progress bottom up on the call graphs of $P_1,P_2$, abstract those functions that were already proven to be equivalent with uninterpreted functions, turn loops into recursion, and abstract the recursive calls also with uninterpreted functions. This enables it to create verification conditions in the form of small programs that are loop- and recursion-free. This method works well for recursive functions as long as they are in sync, and typically fails otherwise. In this work we study the problem of proving equivalence when the two recursive functions are not in sync. Effectively we extend  previous work that studied this problem for functions with a single recursive call site, to the general case. We also introduce a method for detecting automatically the unrolling that is necessary for making two recursive functions synchronize, when possible. We show examples of pairs of functions with multiple recursive calls that can now be proven equivalent with our method, but cannot be proven equivalent with any other automated verification system.
\end{abstract}

\section{Introduction: Regression Verification}
It is very common that changes in a program introduce new unwanted behaviors (`bugs'). Informal methods for checking the new code include code review, testing the new program, and \emph{regression testing}~\cite{DBLP:conf/esec/Zeller99,D02}, a method by which the output of the new and old versions of the program are compared on a given set of test cases. Formal methods include program verification, which attempts to formally verify the correctness of the new program against user-specified assertions in the code, and formal program equivalence checking. Both problems are undecidable in general. While the general problem of proving equivalence of programs was studied for decades, mostly in the theorem proving world, e.g.,~\cite{C02,MK02,MV06}, the specialization of this problem to closely-related programs, coined by  
Godlin and Strichman as \emph{regression verification}~\cite{GS06,GS07,DBLP:conf/dac/GodlinS09}, stirred the development of various tools and approaches, such as \tool{RVT}~\cite{DBLP:conf/dac/GodlinS09}, \tool{Reve}~\cite{DBLP:conf/kbse/FelsingGKRU14}, \tool{SymDiff}~\cite{DBLP:conf/cav/LahiriHKR12}, \tool{ARDiff}~\cite{10.1145/3368089.3409757} and others. Regression verification suggests verifying the program against an older version of itself rather than verifying it against a user-supplied specification or assertions. That is, generate an automated proof of equivalence between a program and its previous version. Although this problem is still undecidable, for closely related programs regression verification is believed to be easier in practice than program verification, and also spares the need to manually specify assertions. Furthermore, it can exploit the similarity between the programs to enable the proof and to reduce the computational complexity~\cite{DBLP:conf/dac/GodlinS09}. Admittedly the notion of correctness is weaker compared to functional verification, since the previous version of the code may have suffered from the same problems as the new one, but on the other hand it might be applicable in large scale programs, where formal functional verification is impractical, and where there is no specification or the assertions cover only a small part of the program's behavior. Regression verification is relevant especially when the code has changed but is supposed to behave the same towards its interface, e.g., as a result of \emph{refactoring} and \emph{performance tuning}. There is more than one definition to what equivalence between functions means (see six definitions in~\cite{GS08}), but here we will use the most commonly used, namely that of \emph{partial equivalence}~\cite{DBLP:conf/dac/GodlinS09}: $f_1$ and $f_2$ are partially equivalent if given the same input, and both terminate, they return the same output. 
	
The problem of program equivalence can be easily reduced to that of software verification, by generating a new program, whose template is shown in Fig.~\ref{fig:rvtmainprogram} for a case in which the compared main functions are $\fo$ and $\ft$ (these functions can generally call other functions). This program executes the given programs with the same non-deterministic input and then asserts that their output is equal (renaming of functions may be required to avoid collisions, and global variables need to be added to the parameters lists). Hence, 
any software verification tool as well as interactive theorem provers can be used to prove equivalence between programs.  However, this route does not take any advantage of the assumed similarity between the compared programs. 
	
	\begin{figure}
	\begin{center}
%		\begin{minipage}{7 cm}			
\begin{lstlisting}
<code of f1(i) and f2(i)>;

main() {
  i = non_det();
  assert(f1(i) == f2(i));
}
\end{lstlisting}
%		\end{minipage}
		\caption{A template for proving the equivalence of two functions, $\fo$ and $\ft$, by a reduction to a single-program verification task. Assume that the input {\tt i} is sent by value. }
		\label{fig:rvtmainprogram}
	\end{center}
\end{figure}

Let us mention here a selected set of systems that specialize in regression verification. \tool{RVT} (Regression Verification Tool) will be described in detail in the next section, as its strategy for proving the equivalence of recursive functions is the basis of the current work. \tool{Reve} (stands for REgression VErification) \cite{DBLP:conf/kbse/FelsingGKRU14} uses \emph{weakest liberal precondition} ($wlp$) calculus to reduce the equivalence problem into a Horn constraint and solve it using Z3 \cite{DBLP:conf/sat/HoderB12} or Eldarica \cite{DBLP:conf/cav/RummerHK13}, to try to prove the equivalence or generate a counterexample.
\tool{Symdiff} by Microsoft \cite{DBLP:conf/cav/LahiriHKR12} translates the given programs into an intermediate language called Boogie, which has its own formal verifier based on \tool{Z3}.
\tool{ARDiff} by Badihi at al.~\cite{10.1145/3368089.3409757} is, to the best of our knowledge, the latest implementation of a regression verification tool and is based on symbolic execution. It also uses an abstraction/refinement loop based on CEGAR~\cite{10.1007/10722167_15} to remove parts of the code that are not relevant for the equivalence proof.
Noteworthy previous tools in this area are \tool{CLEVER}~\cite{9285657}, \tool{MDDiff}~\cite{10.1007/978-3-319-66706-5_20}, \tool{IMP-S}~\cite{inproceedings} and \tool{DSE}~\cite{10.1145/1453101.1453131}. 

The topic of this article is a new proof rule and a corresponding (incomplete) verification method for proving the equivalence of unbalanced recursive functions, when there are multiple such calls in the top frame of the functions. The simpler case, of a single recursive call, was dealt with in the past in~\cite{DBLP:conf/fm/StrichmanV16}. We also show a method for finding the proper unrolling of the two compared functions so they will synchronize. 
Our implementation in \tool{RVT} of this rule together with the automation of the unrolling, is capable of proving equivalence of functions that no other regression verification tool or software verification tool can prove. 

The rest of the article is structured as follows. In the next section we explain RVT's strategy of proving the equivalence of recursive functions. In Sec.~\ref{sec:equiv} we show why proving the equivalence of recursive functions that are not synchronized, fail the proof rule that is used by RVT. Our contribution begins in Sec.~\ref{sec:equiv}, where we show how to prove equivalence of unbalanced recursive functions when there is a single path to the function call site, given a correct unrolling of the functions that make them synchronize. In Sec.~\ref{sec:synch} we show how to automate the process of finding that correct unrolling. In Sec.~\ref{sec:multistep} we generalize the results of Sec.~\ref{sec:equiv} to general functions, i.e., with multiple paths, by leveraging `concolic execution' to get all the paths of the function in its top frame, to recursive call sites. We conclude that section by describing cases that demonstrate the incompleteness of our new proof rule.

\section{An overview of RVT's Strategy for Proving Equivalence}
\label{sec:rvtreview}
%RVT is a Regression Verification Tool, originally developed by Godlin and Strichman . 
Given a mapping between the functions of two programs, RVT~\cite{DBLP:conf/dac/GodlinS09,DBLP:journals/stvr/GodlinS13} tries to prove the partial equivalence of each of the paired functions. It has a decomposition algorithm, which traverses bottom-up the call graphs of the two programs, and generates each time a verification condition for proving the equivalence of a pair of functions. If that verification condition is proven correct, then RVT marks those functions as equivalent, and from thereon any call to those functions is replaced with a call to the same uninterpreted function. This algorithm becomes more complicated in the presence of mutually recursive functions, but we will not go into more details here. 
The full decomposition algorithm is described as Alg 1 in~\cite{DBLP:journals/stvr/GodlinS13}. 

In the rest of this section we explain how RVT attempts to prove the equivalence of two recursive \emph{functions}, rather than full programs. Furthermore, we will assume from hereon that those functions accept a single parameter by value, that there are no global variables, and that if those functions call other functions, then those were already proven equivalent and replaced with uninterpreted functions.   

RVT's strategy for proving the equivalence of recursive functions is inspired by Hoare's inference rule for proving a given pre-post condition over a single recursive function, that was published already in 1971~\cite{DBLP:series/lnm/Hoare71}: 

\begin{equation} \label{eqn:parteq}
	{\frac {\{p\} call\ proc \{q\} \vdash_H \{p\} proc\text{-}body \{q\}}{\{p\} call\ proc \{q\}}} 
\end{equation}
Here $p$ and $q$ are the pre- and post-conditions. The rule says that if assuming that the $p,q$ relation holds for calls to procedure $proc$ enables us to prove that this relation holds for the procedure's body (\emph{proc-body}), then it also holds for calls to $proc$. Hence, the correctness is based on an inductive argument: we use what we want to infer as a hypothesis on the recursive call. 

\cite{DBLP:conf/dac/GodlinS09} introduced a new inference rule for proving the equivalence of two recursive functions, that is based on the same principle. Here, however, the pre- and post-conditions are fixed: the pre-condition is the equivalence of the inputs, and the post-condition is the equivalence of the outputs. More specifically, for recursive functions $\fo$ and $\ft$ the rule is
\begin{equation}\label{eq:rvtrule}
	{\frac {
			\begin{array}{c}
			{\emph{in}[call\ \fo]=\emph{in}[call\ \ft]\rightarrow \emph{out}[call\ \fo]=\emph{out}[call\ \ft]}\\
			 \vdash{\emph{in}[\fo\ body]=\emph{in}[\ft\ body]\rightarrow \emph{out}[\fo\ body]=\emph{out}[\ft\ body]}
	\end{array}	 
	  }
		{{\emph{in}[call\ \fo]=\emph{in}[call\ \ft]\rightarrow \emph{out}[call\ \fo]=\emph{out}[call\ \ft]}}} 
		\rule{PART-EQ}
\end{equation}
This rule states that if assuming that the calls to $\fo$ and $\ft$ are partially-equivalent (i.e., same inputs result in the same outputs) enables us to prove that $\fo$ and $\ft$ bodies are partially-equivalent, then $\fo$ and $\ft$ are partially equivalent. RVT checks the premise of (\ref{eq:rvtrule}) by replacing the recursive calls in $\fo,\ft$ with calls to the same uninterrupted function ($UFs$). An uninterpreted function, by definition, is congruent and over-approximates the substituted function. Hence, this substitution fulfills the first part of the premise. We define \emph{related calls} as:
\begin{definition}[Related Calls]
	\label{def:relatedcalls}
	Recursive calls to $\fo$ and $\ft$ are called \emph{related} if they are made with the same parameters.  
\end{definition}
The uninterpreted functions that replace the recursive calls emit the same (nondeterministic) value on related calls, which is many time sufficient for proving equivalence. 

\begin{example} \label{ex:main}
Consider the two functions in  Fig.~\ref{fig:sump}. 
The verification condition that RVT generates uses the template of Fig.~\ref{fig:rvtmainprogram}, where the two functions in Fig.~\ref{fig:sump} are replaced with those that appear in Fig.~\ref{fig:sump:vc}. A `pass' result on this condition implies that the  premise of (\ref{eq:rvtrule}) is true. 
\begin{figure}
\begin{center}
\begin{tabular}{ll}
\begin{minipage}{6 cm}
\begin{lstlisting}
int sum1(int n){
   if (n <= 1) return n;
   return n + sum1(n-1);
}
\end{lstlisting}
\end{minipage} & 
\begin{minipage}{6 cm}
\begin{lstlisting}
int sum2(int n){
   int res;
   if (n <= 1) return n;
   res = sum2(n-1);
   return res + n;
}
\end{lstlisting}
\end{minipage}
\end{tabular}
\caption{Two equivalent recursive functions.}\label{fig:sump}
\label{fig:sum}
\end{center}
\end{figure}

\begin{figure}
\begin{center}
\begin{tabular}{ll}
\begin{minipage}{6 cm}
\begin{lstlisting}
int sum1(int n){
   if (n <= 1) return n;
   return n + UF(n-1);
}
\end{lstlisting}
\end{minipage} & 
\begin{minipage}{6 cm}
\begin{lstlisting}
int sum2(int n){
   int res;
   if (n <= 1) return n;
   res = UF(n-1);
   return res + n;
}
\end{lstlisting}
\end{minipage}
\end{tabular}
\caption{The functions of  Fig.~\ref{fig:sump}, 
after replacing the recursive calls with the \emph{same} uninterpreted function `UF'}.
\label{fig:sump:vc}
\end{center}
\end{figure}
\end{example}

As a preprocessing step, RVT replaces all loops with calls to new recursive functions. Hence after replacing all the recursive calls with uninterpreted functions, the program in Fig.~\ref{fig:rvtmainprogram} is \emph{flat}, and hence its verification is decidable. RVT calls CBMC~\cite{ckl2004}, a software bounded model checker, to verify this program.

\subsection{Proving equivalence of functions with Unbalanced Recursion }
\label{sec:equiv}

Example ~\ref{ex:main} listed functions that are in \emph{sync}, which we formally define as follows.
\begin{definition}[Sync] \label{def:sync}
	Let $\fo$ and $\ft$ be two recursive functions. $\fo$ and $\ft$ are said to be in \emph{sync} if for every input they are called recursively at the top frame with the same sets of parameters, or not called at all. 
\end{definition}
%
%We emphasize that this definition does \emph{not} require that the number of recursive calls is the same on both sides for a given input -- it only considers the parameter \emph{sets}, e.g., it is possible that when $\fo$ is called with an input $in$, it calls itself recursively twice at the top frame with a parameter value $in'$, whereas $\ft$ calls itself only once with a parameter $in'$.
In such cases \rule{part-EQ} is typically very successful in proving equivalence. 
However, it turns out that once the functions are not in sync, it typically fails. 
In such a case the uninterpreted functions do not necessarily return the same value, because they are called with different parameters, and this typically fails the proof of equivalence. Moreover, if there exists an input for which $\fo$ invokes the base case whereas $\ft$ makes a recursive call, \rule{part-EQ} fails as well, because the recursive call is replaced with a UF which is not matched on the other side, and that UF returns a non-deterministic value. As mentioned in the introduction, a past publication~\cite{DBLP:conf/fm/StrichmanV16} attempted to deal with this problem, but only considered a simple case in which there is a single recursive call in the function, and also relied on a manual input of the unrolling that is necessary in order to synchronize the two functions. 
The rest of this article is dedicated to a proof rule and corresponding strategy that in some cases are able to solve the general case in which there are multiple calls and some of them may be conditioned, and do so automatically.

\section{Equivalence of \emph{single-path} unbalanced recursions\Long{: the rule \rule{PATH-PART-EQ}} }
\label{sec:adaptstep}

\begin{figure}
\begin{center}
\begin{tabular}{ll}
\begin{minipage}{6.5 cm}
\begin{lstlisting}
int f1(int n){
   if (n < 1) return 0;
   if (n <= 2) return 1; 
   return f1(n-1) + f1(n-2);
}
\end{lstlisting} 
\end{minipage} 
&
\begin{minipage}{7 cm}
\begin{lstlisting}
int f2(int n){
   if (n < 1) return 0;
   if (n <= 2) return 1;
   return f2(n-2) + f2(n-2) 
                  + f2(n-3) ;
}
\end{lstlisting}
\end{minipage}
\end{tabular}
\caption{Two equivalent implementations of the Fibonacci sequence.}
\label{fig:f1f2}
\end{center}
\end{figure}

In this section we will consider pairs of recursive functions $\fo,\ft$ with the following properties: 
\begin{itemize}
    \item the control flow of $\fo$ and $\ft$ contains a single path with recursive calls (in Sec.\ref{sec:multistep} 
    we will remove this requirement)
    %(in Sec.~\ref{sec:multistep} we will remove this requirement), 
    \item $\fo$ and $\ft$ are not in sync, and
    \item there exists an unrolling of $\fo,\ft$ that synchronizes them.
\end{itemize} 

Let us demonstrate a case in which two input functions $\fo$ and $\ft$, which are not in sync, can be synchronized with unrolling. 
\begin{example} \label{ex:unroll}
Consider the pair $\fo,\ft$ in Fig~\ref{fig:f1f2}. We will use them as a running example in this article. 
These functions are not in sync: for example, with input $n=4$, $\fo$ makes recursive calls with parameters $3,2$ whereas $\ft$ makes recursive calls with parameters $2,1$. Hence, e.g., $\fo$'s recursive call with the parameter 3 does not have a related call (see Def.~\ref{def:relatedcalls}) by $\ft$, which fails a proof based on the technique that was described in Sec.~\ref{sec:rvtreview}. 
However, after unrolling the call $\fo(n-1)$ once, each recursive call has a related call, as depicted in Fig.~\ref{fig:f1f2cgunrolled}.
\qed 
\end{example}

More formally, we need to find an unrolling of the two functions with the following property. 
\begin{definition}[Sync-Unrolling] \label{def:Sync-Unrolling}
	An unrolling \emph{su} of two recursive functions $\fo$ and $\ft$ is called \emph{Sync-Unrolling} if after applying it the sets of recursive calls on both sides are the same.
\end{definition}
Going back to Example~\ref{ex:unroll}, the unrolling that we applied is indeed a sync-unrolling, because both functions now have the set of parameters $\{n-2,n-3\}$. 

Since the recursive calls in the functions can be conditioned, the correct unrolling may depend on the input value, and in such cases it is not likely that there is a single unrolling which is a sync-unrolling for all inputs. This is the reason that we decompose the proof to the granularity of paths in the top frame of the functions, and the reason that in this section we only consider a single path. 

Also, clearly such an unrolling does not necessarily exist even if the two functions are equivalent. For example, if $\fo$ sums an array of numbers from the beginning to the end, and $\ft$ in the reverse order, no unrolling can synchronize these two functions.

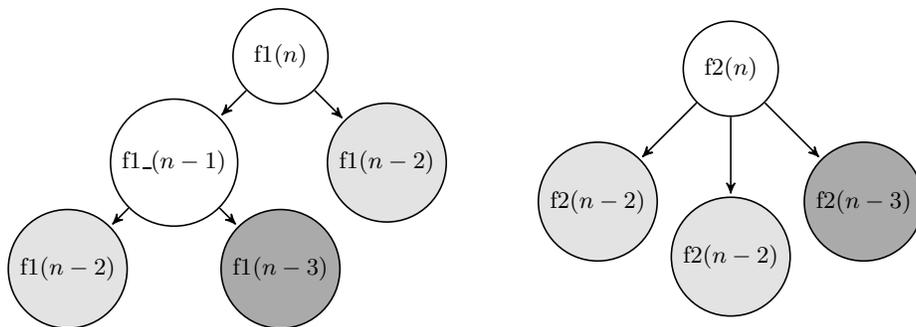
\begin{figure}
\begin{minipage}{.15\textwidth}
\begin{tikzpicture}[>=stealth',shorten >=1pt,node distance=2.0cm,on grid,semithick]
\tikzstyle{n2} = [circle,draw=black,fill={rgb:black,1;white,8}]
\tikzstyle{n3} = [circle,draw=black,fill={rgb:black,3;white,6}]
    \node[state] (f1n) {$\ \fo(n)\ $};
    \node[state] (f1n1) [below left =of f1n] {$\fo\_(n-1)$};
    \node[n2] (f1n2) [below right =of f1n] {$\fo(n-2)$};
    \node[n2] (f1n22) [below left =of f1n1] {$\fo(n-2)$};
    \node[n3] (f1n3) [below right =of f1n1] {$\fo(n-3)$};
    
    \tikzset{mystyle/.style={->}}
    \tikzset{every node/.style={fill=white}}
    \draw (f1n) edge [mystyle]  (f1n1);
    \draw (f1n) edge [mystyle]  (f1n2);
    \draw (f1n1) edge [mystyle]  (f1n22);
    \draw (f1n1) edge [mystyle]  (f1n3);
\end{tikzpicture}
\end{minipage}
\hspace{5cm}
\begin{minipage}{.15\textwidth}
\begin{tikzpicture}[>=stealth',shorten >=1pt,node distance=2.5cm,on grid,semithick ]
\tikzstyle{n2} = [circle,draw=black,fill={rgb:black,1;white,8}]
\tikzstyle{n3} = [circle,draw=black,fill={rgb:black,3;white,6}]
    \node[state] (f1n) {$\ \ft(n)\ $};
    \node[n2] (f1n2) [below  =of f1n] {$\ft(n-2)$};
    \node[n2] (f1n22) [below left =of f1n] {$\ft(n-2)$};
    \node[n3] (f1n3) [below right =of f1n] {$\ft(n-3)$};
    
    \tikzset{mystyle/.style={->}}
    \tikzset{every node/.style={fill=white}}
    \draw (f1n) edge [mystyle]  (f1n2);
    \draw (f1n) edge [mystyle]  (f1n22);
    \draw (f1n) edge [mystyle]  (f1n3);
\end{tikzpicture}
\end{minipage}
\caption{The call graphs of $\fo$ and $\ft$ after unrolling $\fo(n-1)$ once. Related calls are colored the same.}
\label{fig:f1f2cgunrolled}
\end{figure}

We observe that the `pruned' call graph corresponding to the sync-unrolling of a function has a specific structure. Let $c(f)$ denote the number of function calls in the path of $f$ that is being considered. The root of the call graph is $f$, and the number of branches that each node has is either $c(f)$ (if we choose to expand it), or 0 (if we choose not to expand it). Reconsidering Example~\ref{ex:unroll}, $\fo$ has two recursive calls ($c(\fo)=2$) in its path, and indeed each node has either 2 or 0 branches. In Sec.~\ref{sec:synch} we will describe a method for finding synchronization unrolling instructions automatically, when at all possible. 

Let $su(\fo,\ft)$ denote the \emph{sync-unrolling} instructions of $\fo$ and $\ft$, which we will abbreviate to $su$ from hereon. 
Correspondingly, $\fosu,\ftsu$ denote $\fo,\ft$ after being unrolled according to $su$. If $su$ exists and we can find it, then our proof strategy is based on splitting the input domain to two:
\begin{itemize}
\item The \emph{base-case} domain: all the inputs that reach a base case on at least one of $\fosu$ or $\ftsu$. 
\item The \emph{step-case} domain: the complement of the base-case domain. 
\end{itemize}

We then use the following proof rule:
\begin{equation}
 {\frac {\text{\premise{PATH-BASE-EQUIV}}(\fosu,\ftsu) \quad \text{\premise{PATH-STEP-EQUIV}}(\fosu,\ftsu)}{\text{\premise{PARTIAL-EQUIV}}(\fo,\ft)}} 
  \quad \text{\rule{PATH-PART-EQ}}
\end{equation}
The premises (\premise{PATH-BASE-EQUIV}) and (\premise{PATH-STEP-EQUIV}) prove equivalence for the base-case and step-case domains, respectively. 

\begin{proposition} \label{prop:pathsound}
Rule \rule{PATH-PART-EQ} is sound.
\end{proposition}

\begin{proof}
The whole input domain is covered by definition of the step case as a complement of the base case. Hence if each of (\premise{PATH-BASE-EQUIV}) and (\premise{PATH-STEP-EQUIV}) is a sound test of equivalence for its designated inputs, then the rule is sound. After presenting how we test those premises below, we will prove in Proposition~\ref{prop:premisesound} that those methods are indeed sound as well.
\qed	
\end{proof}

For the step case we use the original method of \tool{RVT} as described in Sec.~\ref{sec:rvtreview}, namely we replace all recursive calls with UFs. Note that since the input functions are possibly unrolled, only the actual recursive calls are replaced, e.g., only the calls at the leaves in Fig.~\ref{fig:f1f2cgunrolled}.  

For the base-case we must use a different strategy, because the above strategy would likely fail in cases where there exists an input that drives $\fosu$ to return on a base case, whereas it drives $\ftsu$ to call an uninterpreted function. As an example, consider the input $n=3$ in Fig.~\ref{fig:f1f2}. In $\fosu$ it returns on the base cases, whereas in $\ftsu$ it returns on a recursive call turned into a call to a UF. That called UF has no related call in $\fosu$ and is hence likely to fail the proof.  

Our strategy for handling the base-case is thus different: we maintain the recursive calls in $\fosu,\ftsu$ (that is, we do not replace them with UFs), and use symbolic execution to attempt to prove the equivalence. Since we restrict the runs to the base-case domain, it is likely that both sides will converge after a few iterations, which will enable the proof. In the next subsection we describe this method in more detail. 

\begin{proposition} \label{prop:premisesound}
The tests of (\premise{PATH-BASE-EQUIV}) and (\premise{PATH-STEP-EQUIV}) are sound.
\end{proposition}
\begin{proof}	
Both rules are checked based on the reduction presented in Fig.~\ref{fig:rvtmainprogram}, of equivalence to verification, a reduction that was used in most past publications about program equivalence, e.g.,~\cite{GS06,GS07}, and its soundness is rather obvious. 
The fact that we check unrolled recursions cannot hinder soundness, because 
unrolling $\fo,\ft$ into $\fosu,\ftsu$ does not change the semantics of the functions. After unrolling, we use symbolic execution in the case of the base-case, which is a well-known sound technique, and the rule \rule{part-eq} (see (\ref{eq:rvtrule})), as described in the previous section, that was proven to be sound in~\cite{GS06}.
\qed
\end{proof}
%Some prior work such as \cite{DBLP:conf/fm/StrichmanV16} tried to cope with this challenge but their solution can prove only a limited subset of functions in comparison to our proposed method.

\subsection{(\premise{PATH-BASE-EQUIV}) - Proving equivalence in the base-case domain}\label{sec:SIMPLE-BASE-EQUIV}
Algorithm~\ref{alg:ExtendedBaseProof} describes our process for proving (\premise{PATH-BASE-EQUIV}) (a note on typography from hereon: when we refer to one of $\fo,\ft$ we write $f_i$ for $i \in\{1,2\}$. We will avoid repeating the range $i \in\{1,2\}$). 
Its parameters are the compared functions and the sync-unrolling $su$ information: a sequence of unrolling instructions corresponding to a pruned call graph as explained in the previous section. The algorithm begins by applying these unrolling instructions, after which the recursive calls at the leaves are in sync.

\noindent
\begin{algorithm}
\begin{minipage}{\linewidth}
\begin{algorithmic}[1]
\Function{ProvePathBaseEquiv}{ Programs $\fo,\ft$, sync-unrolling $su$}
    \For { $ i \in \{1,2\}$}\label{step:foreach_p}
	\State$\bar{f_i}$ = \alg{ApplyUnrolling}($f_i$,$su$)
	\State $bcpc_i$ = \alg{GetBaseCasePrecondition}$(\bar{f_i})$
	\EndFor
	\State $P$ = \alg{CreateVerificationProgram}$(\bar{\fo},\bar{\ft})$
    \State $P$ = \alg{AddAssumption}($P$, $bcpc_1 \lor bcpc_2$) \label{step:assumebcpc12}
    \State $P$ = \alg{AddEquivalenceAssertion}($P$)
    \State return \alg{SymbolicExecution}($P$) \label{step:symex}
	\EndFunction
\end{algorithmic}
\end{minipage}
\caption{Checking (\premise{Path-BASE-EQUIV}).}
\label{alg:ExtendedBaseProof}
\end{algorithm}

It then infers a weakest precondition predicate for the base-case domain of $\bar{f_i}$, denoted $bcpc_i$. Specifically \alg{GetBaseCasePrecondition}$(f)$ finds this predicate in two steps:
\begin{enumerate}
	\item In the first step it finds a condition on the input parameters, for the base-case to be taken in the \emph{top frame} of $f$. This is done by (1) replacing all recursive calls with calls to a function with a single assume(false) statement, and (2) invoking \emph{Concolic Execution}~\cite{10.1145/1321631.1321746} (specifically, PathCrawler~\cite{10.1007/11408901_21}), to generate a predicate over the input parameters that represents the feasible paths, namely the base-cases. As an example, in Fig.~\ref{fig:gbc-1} we present the generated program for $\fo$ from Fig~\ref{fig:f1f2}. Note that since this is a flat program (no loops and recursions) by construction, this step is guaranteed to terminate with a correct answer. Let $\rho_i$ denote this predicate for function $f_i$. For $\fo$ in Fig.~\ref{fig:gbc-1} we have $\rho_1 = n < 3$.
\item In the second step, it computes the weakest precondition predicates $bcpc_i$ for the base-case domain of the unrolled function. 
It declares a global variable $bc\_flag$ (initially set to 0) and use $\rho_i$ in order to set it to 1 if the base case is taken in one of the frames. It then adds an `assume($bc\_flag$)' statement after calling $f_i$, in order to restrict the executions to paths that return on at least one base case. The recursive calls are replaced with a call to a function $ret$ (short for `return') with a single return statement. Using Concolic execution again, it computes $bcpc_i$. 

An example of the second step, referring to the same $\fo$ of Fig~\ref{fig:f1f2}, is presented in Fig~\ref{fig:f1f2bc}. Note how $\rho_1 = n < 3$ is used to set the base-case flag in each of the frames (f1, f1\_). The base-case predicate $bcpc_1$ computed in this step is $n \leq 3$: indeed, in this range the unrolled program returns on at least one base case.  
\end{enumerate} 

\begin{figure}
	\begin{center}
	\begin{tabular}{ll}
\begin{minipage}{6.7 cm}
\begin{lstlisting}
int f1(int n) {
	if (n < 1) return 0;
	if (n <= 2) return 1;
	return block(n-1) +
	         block(n-2);
}
\end{lstlisting}
\end{minipage} &
\begin{minipage}{6.7 cm}
\begin{lstlisting}
int block(int n){            
    assume(false);
}			
\end{lstlisting}
\end{minipage}
\end{tabular}
\caption{The program generated in the first step of \alg{GetBaseCasePrecondition} and sent to \tool{PathCrawler} in order to compute $\rho_i$.}\label{fig:gbc-1}
\end{center}
\end{figure}

\begin{figure}
\begin{center}
	\begin{tabular}{ll}
\begin{minipage}{6.7 cm}
\begin{lstlisting}
bc_flag= 0;

int main(int n){  
   f1(n);
   assume(bc_flag);
}
 
int f1(int n){
   if (n < 3) bc_flag = 1;
   if (n < 1) return 0;
   if (n <= 2) return 1;
   return f1_(n-1) + ret(n-2);
}
\end{lstlisting}
\end{minipage} &
\begin{minipage}{6.7 cm}
\begin{lstlisting}
int f1_(int n){
   if (n < 3) bc_flag = 1;
   if (n < 1) return 0;
   if (n <= 2) return 1;
   return ret(n-1) + ret(n-2);
}

int ret(int n){
    return 0;
}
\end{lstlisting}
\end{minipage}
\end{tabular}
\caption{The program generated in the second step of \alg{GetBaseCasePrecondition} in order to compute the base-case precondition $bcpc$.}
\label{fig:f1f2bc}
\end{center}
\end{figure}

\tool{RVT} calls \alg{CreateVerificationProgram} to generate a verification condition following the pattern in Fig.~\ref{fig:rvtmainprogram} while restricting the inputs, in Line \ref{step:assumebcpc12}, by adding $assume(bcpc_1 \lor bcpc_2)$ at the beginning of $P$. 
Finally, it calls a symbolic execution engine to prove the equivalence. By the very fact that we invoke it on a program with recursion, this makes our method incomplete. While one would expect that when we restrict the executions only to inputs that drive the functions to their base-cases, symbolic execution will be successful in a relatively short amount of time, \Long{the following example shows that this is not always the case.}
\Short{in the long version of this article~\cite{SS22} we show an example that demonstrates why this is not guaranteed.}

%\iffalse %comment out for Long 

\begin{figure}[t]
	\begin{center}
		\begin{tabular}{ll}
			\begin{minipage}[t][5cm]{5.5 cm}
\begin{lstlisting}
int p1(int n, int m){
if (m<1||n<1||m > n)
  return 0;
if (m==1||n==1||m == n)
  return 1;
return p1(n - 1, m - 1) 
  + p1(n -1, m);
}
\end{lstlisting}
			\end{minipage} & 
			\begin{minipage}[t][5cm]{7 cm}
\begin{lstlisting}
int p2(int n, int m){
if (m<1||n<1||m > n)
  return 0;
if (m==1||n==1||m == n)
  return 1;
return p2(n - 1, m - 1)  
  + p2(n - 2, m - 1)  
  + p2(n - 2, m);
}
\end{lstlisting}
			\end{minipage}
		\end{tabular}
		\caption{Equivalent implementations of the Pascal's triangle.}
		\label{fig:pascal}
	\end{center}
\end{figure}

\begin{example}
Consider the pair of functions in Fig.~\ref{fig:pascal}. In this example, the base case contains a value of m that is dependent on n. Specifically, if  m=n-1, the call to p1(n-1,m) terminates in its base case since the condition m==n is invoked. On the other hand the other recursive call, {\tt p1(n - 1, m - 1)}, is repeatedly called until one of the lower bounds of {\tt n} or {\tt m} are reached, a process that is likely to fail symbolic execution for large values of {\tt n} and {\tt m}. Indeed \tool{KLEE} times out when trying to prove the equivalence of this base-case.	
\end{example}

%\fi  %comment out for Long 

A schematic depiction of $P$ is shown in Fig.~\ref{fig:basegapvefprogram}(a). If the verification of $P$ in line \ref{step:symex} is successful, then \tool{RVT} deems the base-case equivalence of the pair as correct. 
\begin{figure} [t]
\begin{center}
\begin{tabular}{cc}
\begin{minipage}{6.5 cm}
\begin{lstlisting}[escapeinside={(*}{*)}]
i = non_det()
assume((*$bcpc_1 \lor bcpc_2$*))
res1 = (*$\bar{\fo}$*)(i)
res2 = (*$\bar{\ft}$*)(i)
assert(res1==res2)
\end{lstlisting}
\end{minipage} &
\begin{minipage}{7 cm}
\begin{lstlisting}[escapeinside={(*}{*)}]
i = non_det()
assume(!((*$bcpc_1 \lor bcpc_2$*)))
res1 = (*$\hat{\fo}$*)(i)
res2 = (*$\hat{\ft}$*)(i)
assert(res1==res2)
\end{lstlisting}
\end{minipage} \\
(a) & (b)
\end{tabular}
\caption{A verification program to prove conditional equivalence for (a) the base cases, and (b) the step case. The difference between ${\bar{f}}$ and ${\hat{f}}$ is that in the former recursive calls are maintained, whereas in the latter they are replaced with UFs.}
\label{fig:basegapvefprogram}
\end{center}
\end{figure}

\Long{
\subsubsection{Comparing methods to prove the base-cases equivalence}
We choose to use Symbolic Execution to prove the base cases equivalence due to complexity. We will compare here the complexity of two alternatives. The first is using \tool{CBMC} and increasing the unrolling factor until the verification is successful or timeout is reached. The second is using Symbolic Execution. 
\tool{CBMC} invokes the SAT solver once after traversing the whole call graph, which is a tree as we unroll the recursions, while symbolic execution runs in a DFS style on the call graph and calls the SAT solver every time it reaches a leaf.

Assume that we try to prove the said equivalence of two programs $P_1$ and $P_2$ and that at least one of them contains $C$, $C>1$, recursive calls in its body that are executed on the same trace. Without loss of generality, assume that the equivalence is proved when $P_1$ is unrolled $k$ times, for some $k \geq 1$. 
If we were to use CBMC, the unrolling would create a program that is translated to a formula that represents  $C^k$ calls. The call graphs in Fig.~\ref{fig:f1f2cgunrolled} illustrate this point: the number of nodes in such a tree is $O(C^k)$. The encoding to a propositional formula creates a different set of variables for every call, and thus the number of variables in this formula is proportional to $C^k$. \tool{CBMC} calls a SAT solver to solve this formula, which in itself is exponential in the number of variables. Hence the complexity for the verification task is exponential in $C^k$, thus double exponential in $k$.

On the other hand, using Symbolic Execution, the size of the formula is proportional to $k$, as it represents only a single trace. Symbolic Execution goes through all paths up to a given bound, and the number of paths is worst-case exponential (this is the known ``path explosion" problem  \cite{10.1007/978-3-540-78800-3_28}). Each of these trace verification tasks is a formula handed to a SAT solver and therefore has an exponential run-time bound. That means that the overall complexity of this method is exponential squared in the program's size. 
Hence using symbolic execution for this task has better worst-case complexity.
}

\subsection{(\premise{PATH-STEP-EQUIV}) - Proving equivalence in the step-case domain}
\label{sec:SIMPLE-STEP-EQUIV}
\noindent
\begin{algorithm}
\begin{minipage}{\linewidth}
\begin{algorithmic}[1]
\Function{ProvePathStepEquiv}{ Programs $\fo,\ft$, sync-unrolling $su$, precondition $ebcp$}
	\State$\hat{\fo}$ = \alg{ApplyUnrolling}($\fo$,$su$)
	\State$\hat{\ft}$ = \alg{ApplyUnrolling}($\ft$,$su$)
	\State $P$ = \alg{CreateVerificationProgram}$(\hat{\fo},\hat{\ft})$
	\State $P$ = \alg{ReplaceRecursionsWithUFs}$(P)$
	\State $P$ = \alg{AddAssumption}($P,!ebcp$)
    \State $P$ = \alg{AddEquivalenceAssertion}($P$)
    \State return \alg{CBMC}($P$)
	\EndFunction
\end{algorithmic}
\end{minipage}
\caption{A sound algorithm to prove equivalence of programs for their base cases.}
\label{alg:StepCaseProof}
\end{algorithm}
The step-case proof covers the inputs that are not contained in the base-case domain. Algorithm \ref{alg:StepCaseProof} describes \tool{RVT}'s process for this proof.
It is different than \alg{ProvePathBaseEquiv} (Fig.~\ref{alg:ExtendedBaseProof}) in that it receives as input the base-case precondition that is generated in algorithm \ref{alg:ExtendedBaseProof} and is aliased here as $ebcp$. Based on the definition of the step-case domain (see Sec.~\ref{sec:adaptstep}), it restricts the input space to its negation. Another difference is that \alg{ReplaceRecursionsWithUFs}  replaces all the recursive calls (those that are left after applying $su$) with UFs. 

Now $\hat{\fo}$ and $\hat{\ft}$ are loop-and recursion-free, and can be verified by any software (bounded or unbounded) model-checker. In our case we chose to work with CBMC, which among other things has native support of uninterpreted functions, a useful feature in our case. $P$ is presented in Fig.~\ref{fig:basegapvefprogram} (b). If \tool{CBMC} validates $P$ then the pair of input functions is declared by \tool{RVT} to be step-case equivalent. 

\section{Automatic synchronization of recursive functions}
\label{sec:synch}
The motivation for automating the process of finding sync-unrolling is clear. Both in~\cite{DBLP:conf/fm/StrichmanV16} and~\cite{DBLP:conf/kbse/FelsingGKRU14}, who considered the problem of unbalanced recursive functions, supplied this data manually.   

We already saw in Example~\ref{ex:unroll} a pair of functions for which there exist sync-unrolling (see Fig.~\ref{fig:f1f2cgunrolled}). The unrolled code that \tool{RVT} generates for that example appears in Fig.~\ref{fig:f1f2unrolled}. 

\begin{figure}[t]
\begin{center}
	\begin{tabular}{ll}
\begin{minipage}{7 cm}
\begin{lstlisting}
int f1_(int n){
   if (n < 1) return 0;
   if (n <= 2) return 1; 
   return f1(n-1) + f1(n-2);
 }
int f1(int n){
   if (n < 1) return 0;
   if (n <= 2) return 1; 
   return f1_(n-1) + f1(n-2);
 }
\end{lstlisting}
\end{minipage} &
\begin{minipage}{7 cm}
\begin{lstlisting}
	
int f2(int n){
   if (n < 1) return 0;
   if (n <= 2) return 1;
   return f2(n-2) + f2(n-2) 
                  + f2(n-3);
  }
\end{lstlisting}
\end{minipage}
\end{tabular}
\caption{Unrolling once the call to $\fo(n-1)$, to achieve synchronization.}
\label{fig:f1f2unrolled}
\end{center}
\end{figure}

Our method for finding sync-unrolling, when possible, is described in algorithm \ref{alg:Findunrolling}. 
It creates a C program with an assertion, such that a counterexample represents an unrolling number per function call site, and those numbers define a sync-unrolling. We exploit \tool{CBMC}'s {\tt non\_det()} instruction, corresponding to a non-deterministic choice, to let the solver decide whether to unroll or not in each call site. The program records the call parameters in cases in which it decides not to unroll --- these will be the leaves of the call graph, where we will call the uninterpreted functions. At the end of the generated program we assert that the set of  recorded parameters of the two programs is different. Hence, from a counterexample that \tool{CBMC} finds, we can extract the sync-unrolling, together with a particular input to the functions. 
A single input is enough since we consider a single path (that has recursive calls) at a time.

The main functions of the algorithm are:
\begin{itemize}
    \item \alg{AddDepthAndCallsTracking}$(f)$ adds the auxiliary infrastructure needed for recording the unrolling choices at each recursive call site. Specifically at each call site that it chooses to stop unrolling (i.e., the leaves), it records the depth {\tt D} and the index of the call site {\tt Site} (i.e., each recursive call site is given a unique index). This data is later extracted from the counterexample in order to produce the sync-unrolling. 
    
    \item \alg{GetNaturalBaseCasePrecondition} is called. It produces $\rho_i$ as described in the first step of \alg{GetBaseCasePrecondition} in section \ref{sec:SIMPLE-BASE-EQUIV}.
    
    \item \alg{AddAssumption}$(f,p)$ assumes the precondition $p$ at the beginning of $f$. 
    
    \item \alg{ApplyNonDeterministicRecording}$(f)$ adds a non-deterministic condition. A `true' choice corresponds to no-more unrolling (i.e,. the leaves of the pruned call graph). In that case the input of the current iteration is recorded and the function returns. 
    
    \item \alg{CreateSyncUnrollingProgram}$(\fo,\ft)$ generates a new program based on the template of~\ref{fig:rvtmainprogram}. After the calls to $\fo$ and $\ft$, it asserts that the sets of recorded inputs from \alg{ApplyNonDeterministicRecording} that contain more than one recording (i.e., went through at least one recursive call) are different.

\item \alg{GenerateSyncUnrolling}$(\fo,\ft,ce)$ extracts the values of \texttt{D} and {\tt Site} from the counterexample in order to  generate the sync-unrolling. 
\end{itemize}  

\noindent
\begin{algorithm}
\begin{minipage}{\linewidth}
\begin{algorithmic}[1]
\Function{FindSyncUnrolling}{Loops Free Programs $\fo,\ft$}
    \For { $ i \in \{1,2\}$} \label{step:setupforsu}
    \State $f_i$ = \alg{AddDepthAndCallsTracking}$(f_i)$\label{step:depth_tracking}
    \State $\rho_i$ = \alg{GetNaturalBaseCasePrecondition}$(f_i)$\label{step:get_bcpc}
    \State $f_i$ = \alg{AddAssumption}$(f_i,!\rho_i)$\label{step:block_bc}
    \State $f_i$ = \alg {ApplyNonDeterministicRecording}$(f_i)$
    \EndFor
    \State $P$ = \alg{CreateSyncUnrollingProgram}$(\fo,\ft)$\label{step:create}
    \For{unwinding factor $uw$ increasing from 1 up until a predefined timeout}
        \State $P_{unwinded}$ = \alg{Unwind}$(P,uw)$
        \If {CBMC$(P_{unwinded})$ results with a counterexample $ce$}
        \State \Return \alg{GenerateSyncUnrolling}$(\fo,\ft,ce)$
        \EndIf
    \EndFor
\EndFunction
\end{algorithmic}
\end{minipage}
\caption{An algorithm for finding a synchronization unrolling, if it exists.}
\label{alg:Findunrolling}
\end{algorithm}

\begin{example}
Reconsider the functions in Fig.~\ref{fig:f1f2}. 
The program $P$ that is generated in line \ref{step:create} is shown in Fig.~\ref{fig:f1f2susetup}.

\begin{figure}[t]
\begin{center}
	\begin{tabular}{ll}
		\begin{minipage}{6.5 cm}
\begin{lstlisting}
int f1(int n, int D, int Site) {
    assume(n > 2);
    if(non_det()){
        Leaves1[Size1++] = n;
        return -1;
    }
    if (n < 1) return 0;
    if (n <= 2) return 1; 
    return f1(n-1,D+1,0) + 
           f1(n-2,D+1,1);
}
\end{lstlisting}
\end{minipage} &
\begin{minipage}{6.5 cm}
\begin{lstlisting}
int f2(int n, int D, int Site) {
    assume(n > 2);
    if(non_det()){
        Leaves2[Size2++] = n;
        return -1;
    }
    if (n < 1) return 0;
    if (n <= 2) return 1;
    return f2(n-2,D+1,0) + 
    f2(n-2,D+1,1) + f2(n-3,D+1,2) ;
}
\end{lstlisting}
\end{minipage} 
\end{tabular}
\begin{minipage}{7 cm}
\begin{lstlisting}
main() {
  i = non_det();
  Leaves1 = Leaves2 = {};
  Size1 = Size2 = 0;
  f1(i,0,-1);
  f2(i,0,-1);
  assume(Size1 > 1 && Size2 > 1);
  assert(Leaves1 != Leaves2); // set comparison
} 
\end{lstlisting}
\end{minipage}

\caption{The program generated in line~\ref{step:create} of Algorithm~\ref{alg:Findunrolling}. The capitalized variables are not part of the original functions $\fo,\ft$. {\tt Leaves} and {\tt Size} are global. 
}
\label{fig:f1f2susetup}
\end{center}
\end{figure}
 
After unwinding $P$ enough times \tool{CBMC} will find a counterexample for any $n$ bigger than 2, from which \alg{GenerateSyncUnrolling} derives the sync-unrolling shown in Fig.~\ref{fig:f1f2cgunrolled}.
\qed
\end{example}

\Short{In the long version of this article~\cite{SS22} we also show  an example of equivalent recursive functions that do not have sync-unrolling at all, and an example where it exists, but our suggested algorithm diverges. These are rather contrived examples, however, and in practice we expect that our method will succeed in finding sync-unrolling. } 

%\iffalse %comment out for Long 
\begin{figure}[t]
	\begin{center}
		\begin{tabular}[t]{ll}
\begin{minipage}[t]{6.5 cm}
\begin{lstlisting}
int m1(int n, int flag) {
if (n < 1) return 0;
if (n == 1) return 1;
return m1(n - 1, !flag) 
	+ m1(n - 2, !flag);
}
\end{lstlisting}
			\end{minipage} 
			&
\begin{minipage}[t]{6 cm}
\begin{lstlisting}
int m2(int n, int mode) {
if (n < 1) return 0;
if (n == 1 || n == 2) 
	return 1;
int results = 0;
if (mode)
  results = m2(n - 2,!mode) 
	+ m2(n - 2,!mode) 
	+ m2(n - 3, !mode);
if (!mode)
  results = m2(n - 1, !mode) 
	+ m2(n - 2, !mode);
return results;
}
\end{lstlisting}
			\end{minipage}
		\end{tabular}
		\caption{Equivalent implementations of Fibonacci. m2 alternates between two paths in the control flow.}
		\label{fig:f1f2switch}
	\end{center}
\end{figure}

\begin{example}
%	[*** TODO: can we prove that there isn't a sync-unrolling ? ****]
Consider the pair of functions in Fig.~\ref{fig:f1f2switch}. As we have seen, the steps are equivalent for $n > 3$, {m2}'s base case handles $n=2$ and thus this pair is equivalent. For the case where {\tt mode} is true, \tool{RVT} generates a sync-unrolling that unrolls the call {\tt m1(n-1,!flag)} once and blocks everything else on both sides. When trying to prove the related step case, \tool{CBMC} fails. This happens because in the unrolled frame the polarity of {\tt flag} is negated to the polarity of {\tt mode} in the related calls due to the odd number of unrolling, and as a consequence the generated UFs are unrelated and fails the proof, even though {\tt flag} does not affect {\tt m1}'s result. 
\end{example}

\begin{example} 
In this example we will show a case in which the method we described fails. 
Consider the pair of functions in Fig.~\ref{fig:redundentcalls}. In this example, the values in $\ttwo$ are calculated in advance, and are used differently according to the parity of $n$. The problem occurs when \tool{RVT} tries to generate a sync unrolling while one of the recursive calls is redundant (i.e., its result is not used). Algorithm \ref{alg:Findunrolling} aims to find a value for which all the recursive calls have related calls. To achieve that, each recursive call on each side is either recorded, unrolled or is not reached. The redundant calls in $\ttwo$ are not on a branched path and therefore must be executed; it cannot be recorded as there is no equivalent recording for them in $\tone$ and unrolling them will yield more redundant calls. Whenever \tool{RVT} fails to generate a sync unrolling, it does not unroll at all, and this  fails the proof with this pair.

\begin{figure}
\begin{center}
\begin{tabular}{ll}
\begin{minipage}[t][5cm]{5.5 cm}
\begin{lstlisting}
int t1(int n) {
  if (n < 1) return 0;
  if (n == 1) return 1;
  return t1(n - 1) 
    + t1(n - 2);
}
\end{lstlisting}
\end{minipage} & 
\begin{minipage}[t][5cm]{7 cm}
\begin{lstlisting}
int t2(int n) {
  if (n < 1) return 0;
  if (n == 1 || n == 2) return 1; 
  int results = 0;
  int r1 = t2(n-1);
  int r2 = t2(n-2);
  int r3 = t2(n-3);
  if (n % 2 == 0)
    results = r2 + r2 + r3;
  else results = r1 + r2;
  return results;
}
\end{lstlisting}
				\end{minipage}
			\end{tabular}
			\caption{Equivalent implementations of Fibonacci. $\ttwo$ chooses between two paths depending on the input, but nevertheless makes the same recursive calls.}
			\label{fig:redundentcalls}
		\end{center}
	\end{figure}
	
\end{example}

%\fi  %comment out for Long 

\section{Equivalence of general unbalanced recursions\Long{: the rule \rule{FULL-PART-EQ} }}
\label{sec:multistep}
We now generalize the solution of Sec.~\ref{sec:SIMPLE-BASE-EQUIV} to general unbalanced recursive functions, i.e., with multiple paths in which a recursive call is made. 
As explained right after Def.~\ref{def:Sync-Unrolling}, the correct sync-unrolling may depend on the input, if it leads to different paths. The recursions in Fig.~\ref{fig:f1f2cond} demonstrate this fact.

\begin{figure}
\begin{center}
	\begin{tabular}{ll}
\begin{minipage}[t][3.5cm]{6 cm}
\begin{lstlisting}
int h1(int n){
   if (n < 1) return 0;
   if (n == 1) return 1; 
   return h1(n-1) + h1(n-2);
}
\end{lstlisting}
\end{minipage} &
\begin{minipage}[t][3.5cm]{6 cm}
\begin{lstlisting}
int h2(int n){
   if (n < 1) return 0;
   if (n == 1) return 1; 
   if (n & 1 == 0)
     return h2(n-1) + h2(n-2);
   if (n & 1 == 1)
     return h2(n-2) + h2(n-2) 
                    + h2(n-3);
}
\end{lstlisting}
\end{minipage}
\end{tabular}
\caption{Equivalent implementations of Fibonacci, where h2 has a conditioned step. The condition $n \& 1$ checks the parity of $n$.}
\label{fig:f1f2cond}
\end{center}
\end{figure}

Indeed, for e.g., $n=2$ and $n=3$ the sync-unrolling is clearly different owing to the condition that checks the parity of $n$. Hence, we need to decompose the proof, by partitioning the domain such that each partition induces a single path in the program. 
%A better solution would be to take into account only paths that impact the output of the function using data-dependent analysis (as done in IMP-S \cite{inproceedings}), but to the best of our knowledge there is no industrialized tool that offer this capability.
For each such partition, we will use the method suggested in the previous two sections. Formally, let $CF(f)$ be the set of all the paths in the control flow of the top frame of a function $f$. Let $in_B(f)$ be the set of inputs that drive $f$ to its base case, i.e., without a recursive (or mutual recursive) calls. Let $in(p)$ denote all the inputs that their traces go through the path $p$, $f|_a$ denote the function $f$ with an assumption $a$ restricting its inputs,  and $part(\fo,\ft,p)=in_B(\fo) \cup in_B(\ft) \cup in(p)$. `$su$', as in  Sec.~\ref{sec:adaptstep}, is the sync-unrolling information. This time, those instructions are adjusted to the new restricted functions $f_i|_{part(\fo,\ft,p_i)}$.
Based on this notation, we define the proof rule \rule{FULL-PART-EQ}:

\begin{equation}
{\frac{
  {\displaystyle {\bigwedge_{
  			\begin{array}{l}
  			p_1\in CF(\fo) \\
  			p_2\in CF(\ft)
  	\end{array}	
  	}
%  \bigwedge_{p_2\in CF(\ft)}
 }}\splitdfrac{\text{\premise{PATH-BASE-EQUIV}}(\fo|_{part(\fo,\ft,p_1)}^{su},\ft|_{part(\fo,\ft,p_2)}^{su})\ \wedge}{\text{\premise{PATH-STEP-EQUIV}}(\fo|_{part(\fo,\ft,p_1)}^{su},\ft|_{part(\fo,\ft,p_2)}^{su})}}{\text{\premise{PARTIAL-EQUIV}}(\fo,\ft)}}\text{\rule{FULL-PART-EQ}}
\end{equation}
% todo: is there a soundness theorem here ? 
Note that the programs provided to \rule{PATH-BASE-EQUIV} and  \rule{PATH-STEP-EQUIV} are the same, and hence we can prove these premises as described in the previous sections without a change. 
\begin{proposition}
	\rule{FULL-PART-EQ} is sound.
\end{proposition}
\begin{proof}
\rule{FULL-PART-EQ} covers all the possible pairs of paths. For each such pair it checks the same premise as in \rule{PATH-PART-EQ},
the soundness of which was proved in Proposition~\ref{prop:pathsound}. 
\qed 
\end{proof}

\begin{algorithm}
\begin{minipage}{\linewidth}
\begin{algorithmic}[1]
\Function{ProvePartEq}{ Functions $\fo,\ft$}
    \State $AreEquivalent$ = TRUE
    \For { $ i \in \{1,2\}$}
	\State \label{step:getallpaths} $PathsPredicates_i$ = \alg{GetAllPaths}$(f_i)$
	\State \label{step:GetBaseCasePrecondition} $pp_{f_{base}^i}$ = \alg{GetNaturalBaseCasePrecondition}($f_i$) 
	\EndFor
	\For{every pair ($pp_1,pp_2$) in $PathsPredicates_1 \times PathsPredicates_2$}
	\If{ $ pp_1 \cap pp_2 = \emptyset$} \label{step:skip_unfeasible}
	continue
	\EndIf
	\State \label{step:addassumption} $\fo$ = \alg{AddAssumption}($\fo$, $pp_1\lor pp_{f_{base}^1}$) 
	\label{step:assmp1}
	\State $\ft$ = \alg{AddAssumption}($\ft$, $pp_2 \lor pp_{f_{base}^2}$) 
	\label{step:assmp2}
	\State su = \alg{FindSyncUnrolling}($\fo$,$\ft$)
	\State $BaseCaseEquiv$ = \alg{ProvePathBaseEquiv}($\fo,\ft,su$) 
	\State \label{step:provestepincs} $StepCaseEquiv$ = 	\alg{ProvePathStepEquiv}($\fo,\ft,su,bcpc_1\lor bcpc_2$)
 \If{ \label{step:false}
	!( BaseCaseEquiv $\wedge$ StepCaseEquiv)}
	\Return FALSE; 
	\EndIf
	\EndFor
    \State \Return TRUE;
	\EndFunction
\end{algorithmic}
\end{minipage}
\caption{Proving equivalence of general (i.e., multi-path) unbalanced recursions.}
\label{alg:provemcr}
\end{algorithm}

We check the premise of \rule{FULL-PART-EQ} with Algorithm~\ref{alg:provemcr}.  It partitions the input domain of the two input functions, such that each partition corresponds to a separate path in their flattened version, as explained above. It then tries to prove the equivalence of each pair of such partitions. We now describe the main elements of the algorithm: 
\begin{itemize}
    \item In line~\ref{step:getallpaths},
    \alg{GetAllPaths($f$)} uses \emph{concolic execution} to get all the paths in the top frame of $f$. To that end, \alg{GetAllPaths} flattens $f$ by replacing recursive calls with calls to $ret$, which is a function containing a single $return$ statement and nothing else as described in Fig.~\ref{fig:f1f2bc}, and calls PathCrawler \cite{10.1007/11408901_21}. PathCrawler generates a list of path conditions for all the possible paths in that flattened version of $f$.  
    \item In line~\ref{step:GetBaseCasePrecondition} \alg{GetNaturalBaseCasePrecondition} is called (see Sec.~\ref{sec:synch}). 
    \item In line \ref{step:skip_unfeasible} we screen out infeasible paths, as a performance optimization. It is based on the observation that some of the paths predicates combinations in $PathsPredicates_1 \times PathsPredicates_2$ are not feasible. \Short{More details appear in the long version of this article~\cite{SS22}.}
    \Long{\tool{RVT} checks this by creating a verification task as presented in Fig.~\ref{fig:checkfeasibility}. In the generated program, we assume both path predicates on the same non deterministic inputs. Sending this task to a model checker results in two option: $1.$ If the verification is successful, that means the conjunctions of the assumptions is unsatisfiable and thus the intersection of $pp_1$ and $pp_2$ is empty. $2.$ If the verification fails, that means that the paths in question are feasible as the verification reaches the false assertion.}
    
    \item As of line~\ref{step:addassumption}, for each feasible pair of paths predicates, \tool{RVT} treats the programs as though they had a single step by restricting the inputs and tries to prove them with \alg{ProvePathBaseEquiv} and \alg{ProvePathStepEquiv} from sections \ref{sec:SIMPLE-BASE-EQUIV} and \ref{sec:SIMPLE-STEP-EQUIV}, respectively. \alg{AddAssumption} in those algorithms deviates from how is was originally described. Rather than just adding the assumptions, it replaces the assumptions added in algorithm \ref{alg:provemcr} in lines~\ref{step:assmp1} and~\ref{step:assmp2}. Otherwise, when blocking the base cases for example, we will assume contradicting assumptions and the verification task would be vacuously true. The formula $bcpc_1\lor bcpc_2$ in line \ref{step:provestepincs} is computed as part of  \alg{ProvePathBaseEquiv} in the line above. 
    
    \item In line~\ref{step:false}, we return False (i.e., we failed proving that $\fo$ and $\ft$ are equivalent), if we failed proving equivalence for one of the pairs. 
\end{itemize}

%\iffalse %comment out for Long 
\begin{figure} [t]
\begin{center}
\begin{minipage}{7 cm}
\begin{lstlisting}[escapeinside={(*}{*)}]
i = non_det()
assume(*($pp_1(i)$*) && (*($pp_2(i)$*))
assert(false)
\end{lstlisting}
\end{minipage}
\caption{A verification task for checking   the emptiness of an intersection of two input domains, which are defined by the path predicates $pp_1$ and $pp_2$ over an input argument $i$.}
\label{fig:checkfeasibility}
\end{center}
\end{figure}
%\fi % comment out for Long

\section{Summary}
\rule{FULL-PART-EQ} can prove a range of use cases beyond what was demonstrated here --- see RVT's website: https://rvt.iem.technion.ac.il for more examples. To the best of our knowledge, based on experiments, \tool{RVT} is the only tool that can prove the equivalence of programs with unsynchronized multiple recursive calls such as the pairs in Figs.~\ref{fig:f1f2} and \ref{fig:f1f2cond}. Symbolic execution based tools such as \tool{CLEVER}~\cite{9285657} and \tool{ARDiff}~\cite{10.1145/3368089.3409757} cannot handle unbounded loops and recursions. \tool{REVE}~\cite{FelsingGrebingKlebanov2014} fails to prove those examples. We also created a simple verification task based on the reduction in Fig.~\ref{fig:rvtmainprogram}, and experimented with the state-of-the-art software model checkers that performed best in the 2022 software verification competition SV-COMP 2022 \cite{10.1007/978-3-030-99527-0_20}: \tool{CPAchecker}~\cite{10.1007/978-3-642-22110-1_16} and \tool{Ultimate Automizer}~\cite{10.1007/978-3-319-89963-3_30} failed to prove these examples.

\bibliographystyle{plain}  
\bibliography{references}

\begin{thebibliography}{10}

\bibitem{10.1007/978-3-540-78800-3_28}
Saswat Anand, Patrice Godefroid, and Nikolai Tillmann.
\newblock Demand-driven compositional symbolic execution.
\newblock In C.~R. Ramakrishnan and Jakob Rehof, editors, {\em Tools and
  Algorithms for the Construction and Analysis of Systems}, pages 367--381,
  Berlin, Heidelberg, 2008. Springer Berlin Heidelberg.

\bibitem{inproceedings}
John Backes, Suzette Person, Neha Rungta, and Oksana Tkachuk.
\newblock Regression verification using impact summaries.
\newblock volume 7976, 07 2013.

\bibitem{10.1145/3368089.3409757}
Sahar Badihi, Faridah Akinotcho, Yi~Li, and Julia Rubin.
\newblock Ardiff: Scaling program equivalence checking via iterative
  abstraction and refinement of common code.
\newblock In {\em Proceedings of the 28th ACM Joint Meeting on European
  Software Engineering Conference and Symposium on the Foundations of Software
  Engineering}, ESEC/FSE 2020, page 13–24, New York, NY, USA, 2020.
  Association for Computing Machinery.

\bibitem{10.1007/978-3-030-99527-0_20}
Dirk Beyer.
\newblock Progress on software verification: Sv-comp 2022.
\newblock In Dana Fisman and Grigore Rosu, editors, {\em Tools and Algorithms
  for the Construction and Analysis of Systems}, pages 375--402, Cham, 2022.
  Springer International Publishing.

\bibitem{10.1007/978-3-642-22110-1_16}
Dirk Beyer and M.~Erkan Keremoglu.
\newblock Cpachecker: A tool for configurable software verification.
\newblock In Ganesh Gopalakrishnan and Shaz Qadeer, editors, {\em Computer
  Aided Verification}, pages 184--190, Berlin, Heidelberg, 2011. Springer
  Berlin Heidelberg.

\bibitem{10.1007/10722167_15}
Edmund Clarke, Orna Grumberg, Somesh Jha, Yuan Lu, and Helmut Veith.
\newblock Counterexample-guided abstraction refinement.
\newblock In E.~Allen Emerson and Aravinda~Prasad Sistla, editors, {\em
  Computer Aided Verification}, pages 154--169, Berlin, Heidelberg, 2000.
  Springer Berlin Heidelberg.

\bibitem{ckl2004}
Edmund Clarke, Daniel Kroening, and Flavio Lerda.
\newblock A tool for checking {ANSI-C} programs.
\newblock In Kurt Jensen and Andreas Podelski, editors, {\em Tools and
  Algorithms for the Construction and Analysis of Systems (TACAS 2004)}, volume
  2988 of {\em Lecture Notes in Computer Science}, pages 168--176. Springer,
  2004.

\bibitem{C02}
Sorin Craciunescu.
\newblock Proving the equivalence of clp programs.
\newblock In {\em ICLP}, pages 287--301, 2002.

\bibitem{D02}
Dustin E.
\newblock {\em Effective Software Testing: 50 SpecificWays to Improve Your
  Testing}.
\newblock Addison-Wesley Professional.

\bibitem{DBLP:conf/kbse/FelsingGKRU14}
Dennis Felsing, Sarah Grebing, Vladimir Klebanov, Philipp R{\"{u}}mmer, and
  Mattias Ulbrich.
\newblock Automating regression verification.
\newblock In {\em {ACM/IEEE} International Conference on Automated Software
  Engineering, {ASE} '14, Vasteras, Sweden - September 15 - 19, 2014}, pages
  349--360, 2014.

\bibitem{FelsingGrebingKlebanov2014}
Dennis Felsing, Sarah Grebing, Vladimir Klebanov, Philipp R\"{u}mmer, and
  Mattias Ulbrich.
\newblock Automating regression verification.
\newblock In {\em 29th IEEE/ACM International Conference on Automated Software
  Engineering (ASE 2014)}, ASE '14, pages 349--360. ACM, September 2014.

\bibitem{9285657}
Nick Feng, Federico Mora, Vincent Hui, and Marsha Chechik.
\newblock Scaling client-specific equivalence checking via impact boundary
  search.
\newblock In {\em 2020 35th IEEE/ACM International Conference on Automated
  Software Engineering (ASE)}, pages 734--745, 2020.

\bibitem{GS06}
Benny Godlin and Ofer Strichman.
\newblock An inference rule for the equivalence of procedures.
\newblock (Submitted to Acta-Informatica).

\bibitem{GS07}
Benny Godlin and Ofer Strichman.
\newblock Regression verification - a practical way to verify programs.
\newblock In Bertrand Meyer and Jim Woodcock, editors, {\em VSTTE: Verified
  Software: theories, tools, experiments}, volume 4171, 2005.
\newblock conf in 2005, published in 2007.

\bibitem{GS08}
Benny Godlin and Ofer Strichman.
\newblock Inference rules for proving the equivalence of recursive procedures.
\newblock {\em Acta Informatica}, 45(6):403--439, 2008.

\bibitem{DBLP:conf/dac/GodlinS09}
Benny Godlin and Ofer Strichman.
\newblock Regression verification.
\newblock In {\em Proceedings of the 46th Design Automation Conference, {DAC}
  2009, San Francisco, CA, USA, July 26-31, 2009}, pages 466--471, 2009.

\bibitem{DBLP:journals/stvr/GodlinS13}
Benny Godlin and Ofer Strichman.
\newblock Regression verification: proving the equivalence of similar programs.
\newblock {\em Softw. Test. Verification Reliab.}, 23(3):241--258, 2013.

\bibitem{10.1007/978-3-319-89963-3_30}
Matthias Heizmann, Yu-Fang Chen, Daniel Dietsch, Marius Greitschus, Jochen
  Hoenicke, Yong Li, Alexander Nutz, Betim Musa, Christian Schilling, Tanja
  Schindler, and Andreas Podelski.
\newblock Ultimate automizer and the search for perfect interpolants.
\newblock In Dirk Beyer and Marieke Huisman, editors, {\em Tools and Algorithms
  for the Construction and Analysis of Systems}, pages 447--451, Cham, 2018.
  Springer International Publishing.

\bibitem{DBLP:series/lnm/Hoare71}
C.~A.~R. Hoare.
\newblock Procedures and parameters: An axiomatic approach.
\newblock In {\em Symposium on Semantics of Algorithmic Languages}, pages
  102--116. 1971.

\bibitem{DBLP:conf/sat/HoderB12}
Krystof Hoder and Nikolaj Bj{\o}rner.
\newblock Generalized property directed reachability.
\newblock In {\em Theory and Applications of Satisfiability Testing - {SAT}
  2012 - 15th International Conference, Trento, Italy, June 17-20, 2012.
  Proceedings}, pages 157--171, 2012.

\bibitem{DBLP:conf/cav/LahiriHKR12}
Shuvendu~K. Lahiri, Chris Hawblitzel, Ming Kawaguchi, and Henrique
  Reb{\^{e}}lo.
\newblock {SYMDIFF:} {A} language-agnostic semantic diff tool for imperative
  programs.
\newblock In {\em Computer Aided Verification - 24th International Conference,
  {CAV} 2012, Berkeley, CA, USA, July 7-13, 2012 Proceedings}, pages 712--717,
  2012.

\bibitem{MK02}
Panagiotis Manolios and Matt Kaufmann.
\newblock Adding a total order to acl2.
\newblock In {\em The Third International Workshop on the ACL2 Theorem Prover},
  2002.

\bibitem{MV06}
Panagiotis Manolios and Daron Vroon.
\newblock Ordinal arithmetic: Algorithms and mechanization.
\newblock {\em Journal of Automated Reasoning}, 2006.
\newblock to appear.

\bibitem{10.1145/1453101.1453131}
Suzette Person, Matthew~B. Dwyer, Sebastian Elbaum, and Corina~S.
  Pundefinedsundefinedreanu.
\newblock Differential symbolic execution.
\newblock SIGSOFT '08/FSE-16, page 226–237, New York, NY, USA, 2008.
  Association for Computing Machinery.

\bibitem{DBLP:conf/cav/RummerHK13}
Philipp R{\"{u}}mmer, Hossein Hojjat, and Viktor Kuncak.
\newblock Disjunctive interpolants for horn-clause verification.
\newblock In {\em Computer Aided Verification - 25th International Conference,
  {CAV} 2013, Saint Petersburg, Russia, July 13-19, 2013. Proceedings}, pages
  347--363, 2013.

\bibitem{10.1145/1321631.1321746}
Koushik Sen.
\newblock Concolic testing.
\newblock In {\em Proceedings of the Twenty-Second IEEE/ACM International
  Conference on Automated Software Engineering}, ASE '07, page 571–572, New
  York, NY, USA, 2007. Association for Computing Machinery.

\bibitem{DBLP:conf/fm/StrichmanV16}
Ofer Strichman and Maor Veitsman.
\newblock Regression verification for unbalanced recursive functions.
\newblock In {\em {FM} 2016: Formal Methods - 21st International Symposium,
  Limassol, Cyprus, November 9-11, 2016, Proceedings}, pages 645--658, 2016.

\bibitem{10.1007/978-3-319-66706-5_20}
Anna Trostanetski, Orna Grumberg, and Daniel Kroening.
\newblock Modular demand-driven analysis of semantic difference for program
  versions.
\newblock In Francesco Ranzato, editor, {\em Static Analysis}, pages 405--427,
  Cham, 2017. Springer International Publishing.

\bibitem{10.1007/11408901_21}
Nicky Williams, Bruno Marre, Patricia Mouy, and Muriel Roger.
\newblock Pathcrawler: Automatic generation of path tests by combining static
  and dynamic analysis.
\newblock In Mario Dal~Cin, Mohamed Ka{\^a}niche, and Andr{\'a}s Pataricza,
  editors, {\em Dependable Computing - EDCC 5}, pages 281--292, Berlin,
  Heidelberg, 2005. Springer Berlin Heidelberg.

\bibitem{DBLP:conf/esec/Zeller99}
Andreas Zeller.
\newblock Yesterday, my program worked. today, it does not. why?
\newblock In Oscar Nierstrasz and Michel Lemoine, editors, {\em Software
  Engineering - ESEC/FSE'99, 7th European Software Engineering Conference, Held
  Jointly with the 7th {ACM} {SIGSOFT} Symposium on the Foundations of Software
  Engineering, Toulouse, France, September 1999, Proceedings}, volume 1687 of
  {\em Lecture Notes in Computer Science}, pages 253--267. Springer, 1999.

\end{thebibliography}

\end{document}